\documentclass[twoside, a4paper,10pt]{article}
\usepackage{amssymb}
\usepackage{IEEEtrantools}
\usepackage[mathscr]{eucal}
\usepackage[dvips]{graphicx}
\usepackage{amsmath}
\usepackage{amsthm}
\usepackage{pstricks}
\usepackage{caption}
\usepackage{esint}
\usepackage{cite}
\usepackage{cancel}
\usepackage{tikz,bm}
\usetikzlibrary{arrows}

\usetikzlibrary{arrows.meta}
\def \ni{\noindent}

\newcommand{\be}{\begin{equation}}
\newcommand{\ee}{\end{equation}}
\newcommand{\ben}{\begin{equation*}}
\newcommand{\een}{\end{equation*}}
\newcommand{\bes}{\begin{eqnarray}}
\newcommand{\ees}{\end{eqnarray}}
\newcommand{\besn}{\begin{IEEEeqnarray*}{rCl}}
\newcommand{\eesn}{\end{IEEEeqnarray*}}

\newcommand{\txt}{\textrm}

\newcommand{\tr}{\txt{Tr}}

\newtheorem*{theorem*}{Theorem}

\newtheorem*{definition*}{Definition}
\newtheorem*{lemma*}{Lemma}
\newtheorem*{prop*}{Proposition}
\newtheorem*{corollary*}{Corollary}

\DeclareFontFamily{U}{mathx}{}
\DeclareFontShape{U}{mathx}{m}{n}{<-> mathx10}{}
\DeclareSymbolFont{mathx}{U}{mathx}{m}{n}
\DeclareMathAccent{\widecheck}{0}{mathx}{"71}

\title{A Concentration of Measure Phenomenon in the Principal Chiral Model}
\author{T. Tlas}
\date{}

\begin{document}

\maketitle
\thispagestyle{empty}

\begin{abstract}
\ni We utilize the concentration of measure phenomenon to study the double limit (large-$N$ and continuum) of the $O(N)$ principal chiral model where the large-$N$ limit is performed before the continuum one. The partition function in this double limit is demonstrated to be that of a free massive theory. The result depends crucially on the order in which the limits are taken.
\end{abstract}

\subsection*{Introduction}

The principal chiral model is widely considered to be a simpler version of the Yang-Mills theory which nonetheless has some of its nonperturbative features. The model has been solved \cite{polyakov, faddeev} around 40 years ago using methods which, alas, do not generalize to the more physically relevant case of Yang-Mills. The problem of re-deriving those classical results in an alternative fashion, beginning with the path integral, has remained open. The reason for doing so is not purely academic for it has recently been demonstrated that Yang-Mills theory can be considered to be a principal chiral model in loop space \cite{gaussian}. The aim of this manuscript is to take a step in this direction. Notably, we will obtain the partition function in the large-$N$ limit, which turns out to be that of (an infinite number of copies of) a free, massive scalar field. We will also obtain the mass gap of this limiting free theory explicitly. \newline

 At this point we need to clarify an important and subtle issue regarding the various limits we will be taking. If we label the ultraviolet regulator by $\Lambda$, then the regularized partition function is $Z_{\Lambda, N}[J]$. What we shall do below is obtain the asymptotic form of this function as $N \to \infty$, holding the $\Lambda$ fixed, and take the continuum limit $\Lambda \to \infty$ afterwards. In other words, we will consider the limit
 
 \ben
 \lim_{\Lambda \to \infty} \lim_{N \to \infty} Z_{\Lambda, N} [J].
 \een

This is a similar approach to that in \cite{grosswitten, ym}. It is also consistent with the constructive quantum field theory results obtained in \cite{kopper} where the generated mass at finite $N$ is shown to be in agreement with that obtained below (up to corrections which vanish when $N$ goes to infinity). One could also ask: what is the result of interchanging the order of the two limits? It turns out that it is this other order, usually assumed tacitly, that has been studied exclusively in the literature. It will be out of place to give a full review of the rather extensive and interesting amount of work done on this matter, nonetheless, a very brief summary is warranted. \newline

One approach is to use a form-factor bootstrap method \cite{orland1, orland2} utilizing Smirnov's axioms \cite{smirnov}. This allows one to compute all the correlation functions of the model, which turn out to be \textit{not} those of a free theory despite the fact that the S-matrix is trivial. Another interesting idea is to use the Bethe ansatz by coupling the theory to a conserved charge \cite{fateev, kazakov1} . Curiously, if one chooses this charge cleverly, one also obtains a non-trivial theory which has features reminiscent of a theory of strings \cite{kazakov2}. If on the other hand, one chooses this charge to excite only the lowest mass particles, one does obtain \cite{fateev} a free theory with a similar value for the mass gap obtained in this manuscript. Finally, one can also study the principal chiral model numerically to try to extrapolate to its various limits \cite{campostrini1, campostrini2}. Unfortunately, it is not straightforward to compare the results there with those in this manuscript due to the different methodologies used and the different questions being asked.\footnote{Nonetheless, see the discussion following (7.19) in \cite{campostrini2} which corroborates somewhat that one does obtain a free theory in the large-$N$ limit.}  To sum up, the issue of the different orders of limits is a non-trivial one and it appears that one can obtain different answers depending on how exactly one is taking the different limits. We will come back briefly to these issues at the end of this paper.\footnote{A small caveat to the above discussion is that most of the papers cited deal with the principal chiral model for unitary matrices instead of the orthogonal ones. However, while this changes some of the quantitative answers, qualitatively the models behave in an essentially the same way.} \newline

Let us briefly outline the main ideas of the work below. The general approach could be thought of as a realization of the strategy sketched in \cite{polyakov2}. One starts by swapping the integral over the original field with that over the Lagrange multiplier. Contrary to what happens with the $O(N)$ vector model, one cannot simply apply Laplace's method to the resulting integral. The reason for this is well understood and lies in the fact that, in addition to the action, there is a comparable contribution coming from the entropy \cite{polyakov2, makeenko}, as the number of components of the Lagrange multiplier is large. It has been suggested in \cite{polyakov2} to diagonalize the Lagrange multiplier and then perform the integral over the diagonalizing matrices. Unfortunately, this integral is intractable. We bypass this problem below using the idea of concentration of measure; it shows that in the large-$N$ limit, the effect of the entropy can be modeled by a Gaussian, similarly to how it was shown recently in lattice Yang-Mills \cite{ym}. Using this ansatz, we will be able to rewrite the partition function in a way amenable to the standard asymptotic analysis. Curiously, it turns out that, due to the specific details of the way concentration of measure manifests itself, the entropic fluctuations for the diagonalized field are suppressed in the continuum limit. This is different from the way the same phenomenon appears in lattice Yang-Mills \cite{ym} and will perhaps be the most technically involved aspect of the text below. This suppression easily allows us to find the mass gap of the principal chiral model and explain why it coincides with the result obtained by the, a priori incorrect, naive asymptotic analysis. \newline

The outline of the paper is as follows: in the next section we give the necessary background and describe our setting and conventions. In the subsequent section, we perform the preliminary analysis of the problem so that it's amenable to treatment using the concentration of measure approach. This treatment is started in the section that follows. The necessary calculations of the mean and the variance occupy the section after that. The asymptotic analysis giving the main results is given in the penultimate section. We finish with a brief conclusion and outlook. \newline

\subsection*{Background and Notation}

Let us describe our setting. We will be studying the principal chiral model in the Euclidean setting. We will only deal with the $O(N)$ case, but essentially everything below can be straightforwardly generalized to other groups of interest. The theory is regularized by putting it on a toroidal lattice, with the derivatives replaced by the finite differences. The number of lattice sites will be denoted by $\Lambda^2$. The physical volume of the lattice will be denoted by $V$, and thus each step of the lattice has length $\Delta = \sqrt{\frac{V}{\Lambda^2}    }$. The elements of the lattice will usually be denoted by $x$ (or sometimes $y$). The elements of the (Fourier) dual lattice will be denoted by $p$ (or sometimes $q$). We take the step of the dual lattice to be $\frac{2 \pi}{\sqrt{V}}$. With this choice, the Fourier transform relations for any function $f(x) \leftrightarrow \widehat{f}(p)$ become

\besn
\widehat{f}(p) & = & \sum_x e^{- i px} f(x) \Delta^2 \rightsquigarrow \int e^{- i px} f(x) d^2x\\
f(x) & = & \frac{1}{(2 \pi)^2} \sum_p e^{i p x} \widehat{f}(p) \frac{ (2 \pi)^2}{V} \rightsquigarrow \frac{1}{(2 \pi)^2} \int e^{i px} \widehat{f}(p) d^2p,
\eesn

where the equations on the right are the corresponding continuum limit equations. Note that the continuum limit is achieved by taking $\Lambda \to \infty$, $V \to \infty$ and $\Delta = \frac{V}{\Lambda^2} \to 0$. Since we are interested eventually in the continuum limit of the theory, we will usually use the continuum limit form for most expressions below, as they are often more transparent. We will of course revert back to the explicit discrete notation when issues of regularization become relevant. Note that since $p$ ranges over a square\footnote{We shall take this square to be centered at $(0,0)$. There is no loss of generality in doing so due to periodicity of the Fourier transform.} of side length $ \frac{2 \pi \Lambda}{\sqrt{V}} \equiv \tilde {\Lambda}$, we shall consider our momentum ultraviolet cut-off to be $\tilde{\Lambda}$. \newline

Note that as a consequence of our conventions, if $\widehat{f}(p)$ is rotationally invariant\footnote{More precisely, if it is a restriction of a rotationally invariant function to the lattice.} as a function of $p$, we have the following important formula, which is used repeatedly below

\ben
\sum_p \widehat{f}(p) \rightsquigarrow \frac{V}{ (2 \pi)^2 } \int \widehat{f}(p) d^2p \sim \frac{V}{2 \pi} \int_0^{\tilde{\frac{\Lambda}{2}}} \widehat{f}(p) |p| d|p|   ,
\een

where $\sim$ stands for ``asymptotic to''. The justification for the right hand side is that we've replaced the integral over the square of side length $\tilde{\Lambda}$ with that over the inscribed disc of radius $\frac{\tilde{\Lambda}}{2}$.\footnote{Needless to say, for this to work, one should assume sufficient decay of $\widehat{f}(p)$. The functions we'll be dealing with will satisfy the necessary decay so we will not dwell on this point further.}  \newline

\subsection*{Preliminary Analysis}

The action of the principal chiral model model is

\ben
\frac{N}{2 \lambda} \int \partial_\mu \phi_{ab} \partial_\mu \phi_{ab},
\een

where $\lambda$ is the 't Hooft coupling and the fields $\phi$ satisfy the orthogonality constraint $\phi_{ba}  \phi_{bc} = \delta_{ac}$. As mentioned above, we're giving it in the continuum form. Of course, if we want to be fastidious, we should replace the integral above with a sum over $x$ weighted by a factor of $\Delta^2$, and the derivatives with finite differences. It is undoubtedly clear to the reader that the continuum expression is significantly less cluttered than the discrete one. We will continue to use the continuum expressions freely and will not belabor this issue any more.   \newline

The partition function $Z[J]$ of the theory is given by\footnote{This is of course meant to be regularized, and thus if we are awfully nitpicky, should be written as $Z_{\Lambda, N}[J]$. There will be no loss of clarity in simply ignoring the subscripts in what follows.  }

\ben
 \int \mathcal{D} \phi \mathcal{D} M \exp \bigg [ - \frac{N}{2 \lambda} \int  \partial_\mu \phi_{ab} \partial_\mu \phi_{ab} + i \int M_{ac} \big (\phi_{ba}  \phi_{bc} - \delta_{ac}   \big ) + i \sqrt{\frac{N}{\lambda}} \int J_{ab} \phi_{ab} \bigg ],
\een

where the factor $\sqrt{\frac{N}{\lambda}}$ is inserted in the source term in order to obtain a nontrivial large-$N$ limit. Note that the Lagrange multiplier field satisfies $M_{ac} = M_{ca}$ since the orthogonality constraint is symmetric.   \newline

Rescaling the fields $\sqrt{\frac{N}{\lambda}   } \phi \to \phi$ and $\frac{2\lambda}{N} M \to M$, we see that $Z[J]$ is, up to an irrelevant overall constant, equal to

\besn
\int \mathcal{D} \phi \mathcal{D} M \exp \bigg [ - \frac{1}{2} \int  \partial_\mu \phi_{ab} \partial_\mu \phi_{ab} + \frac{i}{2} \int M_{ac} \big (\phi_{ba}  \phi_{bc} - \frac{N}{\lambda} \delta_{ac}   \big ) + i  \int J_{ab} \phi_{ab} \bigg ]  \\
= \int \mathcal{D} \phi \mathcal{D} M \exp \bigg [ - \frac{1}{2} \int \phi_{ab} K_{ab, a'b'} \phi_{a'b'}    - i \frac{N}{2\lambda} \int M_{aa} + i  \int J_{ab} \phi_{ab} \bigg ] ,
\eesn 

where 

\ben
K_{ab, a'b'} = - \partial^2 \delta_{aa'} \delta_{b b'} - i M_{aa'} \delta_{bb'} = (- \partial^2 \delta_{aa'}  - i M_{aa'} ) \otimes \delta_{bb'} = K_{aa'} \otimes \delta_{bb'}.
\een

We can now perform the integral over the $\phi$'s and get\footnote{Strictly speaking, the formula below is incorrect. This is because there is a factor of $\Delta^2$ in the quadratic term in the $\phi$'s (it comes from the integral) which filters down to a rescaling of the $K$. However, this only amounts to an additive constant in the logarithm term, which contributes an irrelevant overall constant to $Z[J]$ and thus, will be ignored below. Note that the same factor does not cause an issue in the $J$ term. This is because the factor of $\frac{1}{\Delta^2}$ coming from the $K^{-1}_{aa'}$ term cancels with $(\Delta^2)^2$ coming from the $J$'s, leaving a $\Delta^2$ hidden in the expression of the integral above.} 

\besn
Z[J] & = & \int \mathcal{D} M \exp \bigg [  - \frac{N}{2} \tr \Big ( \ln ( K_{aa'}   )  \Big ) - i \frac{N}{2 \lambda} \int M_{aa}   \\
& &   - \frac{1}{2} \int J_{ab} (K_{aa'} )^{-1} J_{a'b}  \bigg ],
\eesn

Now, using the fact that $M$ is a symmetric matrix, we change variables $M = O^t \hat{M} O$, where $\hat{M}$ is diagonal and $O$ is the diagonalizing orthogonal matrix. The Jacobian of this transformation is well-known \cite{mehta} and is equal to $\prod_{a \neq b} | \hat{M}_a - \hat{M}_b|^{\frac{1}{2}} $, where $\{ \hat{M}_a \}_{a = 1, N}$ is the set of eigenvalues of $\hat{M}$. We thus get that

\besn
Z[J] & = & \int \mathcal{D} \hat{M} \mathcal{D} O  \bigg (  \prod_{x, a \neq b} | \hat{M}_a - \hat{M}_b|^{\frac{1}{2}}  \bigg ) \exp \bigg [ - \frac{N}{2} \tr \Big ( \ln ( K_{aa'}  )  \Big )  \\
& &  - i \frac{N}{2 \lambda} \int \sum_a \hat{M}_a - \frac{1}{2} \int J_{ab} (K_{aa'} )^{-1} J_{a'b} \bigg ] ,
\eesn

where the product over $x$ in the Jacobian goes over all the points of space. At this stage, motivated by the form of the propagator $K_{aa'}$, we perform a change of variables $\hat{M} \to \hat{M} + i \mu$ where $\mu > 0$.\footnote{This change of variables can truly be justified only a posteriori once we calculate the large-$N$ limit. Having $\mu >0$ would simplify the analysis by guaranteeing that certain expression are well-defined. Unsurprisingly, it will turn out that $\mu$ and $\hat{M}$ will be related to the mass of the limiting theory.} With this change, we get that $Z[J]$ is equal to

\besn
& & \int \mathcal{D} \hat{M} \mathcal{D} O    \exp \bigg [ \frac{1}{2}    \sum_{x, a \neq b} \ln | \hat{M}_a - \hat{M}_b|  \bigg ] \exp \bigg [ - \frac{N}{2} \tr \Big ( \ln ( K_{aa'}  )  \Big )  \\
& &  - i \frac{N}{2 \lambda} \int \sum_a \hat{M}_a + \frac{N^2V \mu}{2 \lambda}  - \frac{1}{2} \int J_{ab} (K_{aa'} )^{-1} J_{a'b} \bigg ] ,
\eesn

where $V$ above is the volume of space.\newline

Now, as is customary \cite{coleman}, we introduce the probability density

\ben
\rho(\hat{M}) = \frac{1}{N} \sum_a \delta(\hat{M} - \hat{M}_a).
\een

This allows us to rewrite $Z[J]$ as 

\besn
& &\int \mathcal{D} \rho \mathcal{D} O  \exp \bigg [ N^2 \bigg ( \frac{1}{2} \sum_x \int d\hat{M} d\hat{M}'  \rho(\hat{M}) \rho(\hat{M}') \ln | \hat{M} - \hat{M}' |   \\  
& & - \frac{1}{2N} \tr \Big ( \ln ( K_{aa'}  )  \Big ) + \frac{V \mu}{2 \lambda}   -  \frac{i}{2 \lambda} \int \int d \hat{M}  \rho(\hat{M}) \hat{M}  \bigg )- \frac{1}{2} \int J_{ab} (K_{aa'} )^{-1} J_{a'b} \bigg ] .
\eesn

Note that in the penultimate term, one of the integrals is over $\hat{M}$ while the other is over space. Also, observe that $\mathcal{D}\rho$, which stands for a functional integral over $\rho$, is constrained to go over probability measures. \newline

\subsection*{Concentration of Measure}

We would like to obtain the large-$N$ asymptotic of $Z[J]$. To this end, we utilize the concentration of measure phenomenon \cite{ledoux} along the lines demonstrated in \cite{ym}. As a first step, we need the following simple

\begin{lemma*}
The function
\ben
O \to \exp \bigg [ - \frac{1}{2} \int J_{ab} (K_{aa'} )^{-1} J_{a'b} \bigg ] 
\een

is Lipschitz with respect to the metric $d(O, O') = \int || O - O'||_{HS}$ where $|| \cdot ||_{HS}$ stands for the Hilbert-Schmidt norm.
\end{lemma*}

\begin{proof}
Since $e^{-z}$ is Lipschitz as a function of $z$ if $\Re z$ is bounded from below, it is enough to show that $\int J_{ab} (K_{aa'} )^{-1} J_{a'b}$ is Lipschitz as a function of $O$, and that its real part is bounded from below. To this end, consider $K_{aa'} = (- \partial^2 + \mu ) \delta_{aa'} - i M_{aa'}$. Note that the real part of the its eigenvalues is bounded from below by $\mu$. It follows that the real part of the singular values of $K_{aa'}^{-1}$ is contained in $[0, \frac{1}{\mu}]$. From this, two facts follow at once: $\Re ( \int J_{ab} (K_{aa'} )^{-1} J_{a'b}) \geq 0$ and $|| K^{-1}_{aa'} || \leq \frac{1}{\mu}$ where $|| \cdot ||$ is the usual sup norm.\newline

We thus have that

\besn
& & |  \int J_{ab} (K_{aa'} (O) )^{-1} J_{a'b} - \int J_{ab} (K_{aa'} (O') )^{-1} J_{a'b} |   \\
 & \leq & \int | J_{ab} ( K_{aa'}^{-1}(O) - K_{aa'}^{-1} (O')    )   J_{a'b}    | \leq  \int || J ||^2_{HS} || K^{-1}(O) - K^{-1}(O') ||  \\
 &\leq  &  \int || J||^2_{HS} \, \,  \,  ||K^{-1}(O)|| \, \, \, || K^{-1}(O') || \, \,  \,  || K(O) - K(O')||  \\
  & \leq & \int  \frac{|| J ||^2_{HS} }{\mu^2} || O^t \hat{M} O - O'^t \hat{M} O' ||  \\
  & \leq & \int  \frac{|| J ||^2_{HS}}{\mu^2} \Big (  || O^t \hat{M} O - O^t \hat{M} O' ||  +  || O^t \hat{M} O' - O'^t \hat{M} O' ||        \Big )  \\
  & \leq & \int \frac{2 ||J||^2_{HS} || \hat{M}||  }{\mu^2} || O - O'|| \leq \frac{ 2 \max_x ( ||J||^2_{HS} || \hat{M}||   ) }{\mu^2} \int || O - O'||_{HS} ,
\eesn

which concludes the proof.
\end{proof}

Before we proceed, note that we are mainly interested in various derivatives of $Z[J]$ evaluated at 0. The modifications of the proof above needed to cover this case are utterly straightforward and will be left to the reader.\newline

We now use the fact that Lipschitz functions on $O(N)$ concentrate as $N \to \infty$.\footnote{See for example Theorem 5.17 in \cite{meckes}.} This means that such a function converges in probability to its mean, or more intuitively, that the function is essentially constant on its domain. This allows us to replace the term 

\ben
\exp \bigg [ - \frac{1}{2} \int J_{ab} (K_{aa'} )^{-1} J_{a'b} \bigg ] 
\een

in $Z[J]$ with 

\ben
\int \mathcal{D} O \exp \bigg [ - \frac{1}{2} \int J_{ab} (K_{aa'} )^{-1} J_{a'b} \bigg ] ,
\een

which in turn, by applying the same argument to the exponent which has been shown above to be Lipschitz as well, can be replaced with

\ben
\exp \bigg [ - \frac{1}{2} \int \mathcal{D} O  \int J_{ab} (K_{aa'} )^{-1} J_{a'b} \bigg ] .
\een

Before we proceed and calculate this expression, let us note that if the reader is not entirely convinced by the somewhat abstract argument above, we shall sketch below a different, more direct, justification of this manipulation.\newline

In view of the discussion above, we see that in the large-$N$ limit we get the following asymptotic relation

\bes
\label{eqnarray:asymp}
& & Z[J]  \sim \int \mathcal{D} \rho \Bigg ( \exp \bigg [  - \frac{1}{2} \int \mathcal{D} O \int J_{ab} (K_{aa'} )^{-1} J_{a'b} \bigg ] \Bigg )   \\
& &  \times \Bigg ( \exp \bigg [  N^2 \bigg ( \frac{1}{2} \sum_x \int d\hat{M} d\hat{M}'  \rho(\hat{M}) \rho(\hat{M}') \ln | \hat{M} - \hat{M}' | + \frac{V \mu}{2 \lambda}     \nonumber \\
& &  - \frac{i}{2 \lambda} \int \int d \hat{M}  \rho(\hat{M}) \hat{M}  \bigg )   \bigg ] \Bigg )  \Bigg (     \int \mathcal{D} O  \exp \bigg [ N^2 \bigg ( - \frac{1}{2N} \tr \Big ( \ln ( K_{aa'}  )  \Big )  \bigg ) \bigg ] \Bigg ) . \nonumber
\ees

We would now like to handle the last brackets in the expression above. Note that in this case, contrary to what was done in the first term, it is not justified to simply replace the exponent with its average. The reason for this is the presence of the $N^2$ factor in the exponent. Thus, one needs a more refined method to handle this term. We shall proceed by a method similar to that used in \cite{ym} to deal with lattice Yang-Mills. Notably, we would like to pushforward the $\mathcal{D}O$ measure with the function 

\ben
 O \to t(O) = \frac{1}{2N V} \tr \Big (  \ln (K_{aa'})  \Big ),
\een

where $V$ is the volume of space. However, before we do this, we rotate the contour of integration over the $\hat{M}$, or equivalently, that over the $\rho$, so that it is along the imaginary axis. This has the effect of replacing $\hat{M}$ everywhere with $i \hat{M}$, which would guarantee that the integrand in $Z[J]$ above is real. The reason for this contour change is that the reality of the resulting expressions would permit us to use the simple Laplace's method to obtain the asymptotics instead of the saddle point method. \newline

Before we proceed, a couple of quick remarks: first we've done this rotation at this stage and not on the original expression since otherwise, we would not be able to guarantee Lipschitzness of $K_{aa'}$ in the Lemma above. Therefore, we've only rotated the contour after the $e^{- \int J K^{-1} J}$ term is safely taken outside of the $\mathcal{O}$ integral. Second, note that the fact that the contour of integration can be rotated with impunity needs to be justified as the function integrated is not analytic (the culprit being the $|M_a - M_b| = e^{\ln |M_a - M_b|}$ term). It is however a very simple argument to show that this rotation is valid, for the integral we have is a multi-variable generalization of $\int_C f(z) |z| dz$, where $C$ is along the real axis and $f$ is assumed to be analytic and vanishing sufficiently fast at infinity. We can now decompose $C$ as $C_+ + C_-$ where $C_+/C_-$ is the positive/negative part of the real axis. We thus have that 

\besn
\int_C f(z) |z| dz & = & \int_{C_+} f(z) z dz - \int _{C_-} f(z) z dz \\
& = & i^2 \int_{ C_+} f(it) t dt - i^2 \int_{ C_-} f(it) t dt \\
& = & i^2 \int_{- \infty}^\infty f(it) |t| dt.
\eesn

Note that the second equality is obtained by rotating the contours $C_+/C_-$ counter-clockwise by $\frac{\pi}{2}$. This is justified as the two integrands are analytic in the first and third quadrants. \newline

So, having rotated the contour, let us consider the pushforward of $\mathcal{D}O$. Denote the pushforward measure by $d \nu_N[t]$. What can we say about this measure? If we make the reasonable assumption that this measure is asymptotic to $e^{- N^2 f(t) } dt$, with $f(t)$ being sufficiently regular, then in fact, using Laplace's method, we can take $f(t) = \frac{A}{2}(t- t_0)^2$. Moreover, this form of the measure dovetails nicely with the concentration phenomenon. Finally, the coefficients $A, t_0$, and the fact that the overall factor in front of $f$ is $N^2$ (as opposed to some other sequence going to infinity) can be read off from the large-$N$ asymptotics of the first two moments of $d\nu_N[t] dt$.\footnote{Note that one could try to obtain further justification of the asymptotic form above by attempting to compute \textit{all} the moments of $d \nu_N[t] dt$ along the lines that were done in \cite{ym}. However, in this case it is a significantly more involved combinatorial problem and will be left to possible future work.}\newline

We are thus faced with computing the asymptotics of the mean and the variance of $d\nu_N[t] dt$. Since we will be computing various integrals over $O(N)$, let us adapt the graphical notation used in the literature \cite{creutz} to perform such integrals. The Kronecker's delta, $\delta_{aa'}$ will be denoted by an undecorated line, not necessarily straight. Note that it is irrelevant which index is attached to which end of the line. An entry of an element of $O(N)$, $O_{aa'}$ will be denoted by a vertical line decorated by a disk. The upper end of the line represents the index $a$ while the lower represents the index $a'$. Other matrices, e.g. $\hat{M}_{aa'}$ will be denoted by other (non-circular) shapes with two `legs', the left one representing the index $a$ and the right one representing the index $a'$. Connecting two shapes represents identifying the two relevant indices and summing, that is, a product of the relevant matrices. Figure \ref{fig 1} shows a summary of the notation above. \newline

\begin{figure}[h]
\begin{tikzpicture}

\draw (1.5,0) -- (1.5, 1) -- (2, 1) -- (2, 2);
\node at (2.5, 1){$\delta_{aa'}$};

\node at (3.2, 1) {$,$};

\draw (4,0) -- (4, 2);
\draw [fill = white] (4,1) circle [radius=0.2];
\node at (4.3,1.8) {$a$};
\node at (4.35, 0.2) {$a'$};
\node at (4.8, 1) {$O_{aa'}$};

\node at (5.5, 1) {$,$};

\draw (6, 1) -- (6.5, 1.5) -- (7, 1) -- (6,1) ;
\draw (6.2, 1) -- (6.2, 0.5);
\draw (6.8, 1) -- (6.8, 0.5);
\node at (6.2, 0.16) {$a$};
\node at (6.8, 0.2) {$a'$};
\node at (7.5, 1) {$\hat{M}_{aa'}$};

\node at (8.2, 1) {$,$};

\draw (8.7, 1.5) -- (9.2, 2) -- (9.7, 1.5) -- (8.7,1.5 );
\draw (8.9, 1.5) -- (8.9, 0.5);
\draw (9.5, 1.5) -- (9.5, 0.5);
\draw [fill = white] (8.9,1) circle [radius=0.2];
\draw [fill = white] (9.5,1) circle [radius=0.2];
\node at (11.5, 1){$O^t_{aa'} \hat{M}_{a'a''} O_{a''a'''}$};
\node at (8.9, 0.16) {$a$};
\node at (9.65, 0.2) {$a'''$};

\end{tikzpicture}
\caption{\textit{Summary of the graphical notation used below. Note that we've only chosen a particular representation of the Kronecker's delta. We could have drawn any other line. Also note that the last diagram is indeed equal to the displayed mathematical expression since the transpose interchanges the indices.}}
\label{fig 1}
\end{figure}

We need to compute integrals of polynomials over $O(N)$. This is a mature subject (see \cite{weingarten1} for a bird's eye overview) where there are explicit formulas for the expression below. However, since our interest is in the large-$N$ limit, we only need the following simple formula (first obtained in \cite{weingarten2}), giving the dominant asymptotic of a product of an even number of $O$'s: 

\ben
\int dO \, \, O_{a_1 b_1} \dots O_{a_{2k} b_{2k}} = \frac{1}{N^k} \sum \delta_{a_{\alpha_1} a_{\beta_1}} \delta_{b_{\alpha_1} b_{\beta_1}} \dots \delta_{a_{\alpha_k} a_{\beta_k}} \delta_{b_{\alpha_k} b_{\beta_k}} + o( \frac{1}{N^k}), 
\een

where the sum goes over all possible splittings of the collection $\{1, 2, \dots, 2k\}$ into pairs $\{ (\alpha_1, \beta_1), (\alpha_2, \beta_2), \dots, (\alpha_k, \beta_k) \}$. Figure \ref{fig2} gives a graphical representation of the formula above for the $k=1$ and $k=2$ cases.

\begin{figure}[h]
\begin{tikzpicture}

\node at (1,4) {$\int dO$};

\draw (1.7,3) -- (1.7, 5);
\draw [fill = white] (1.7,4) circle [radius=0.2];
\draw (2.2,3) -- (2.2, 5);
\draw [fill = white] (2.2,4) circle [radius=0.2];

\node at (2.9,4) {$= \frac{1}{N}$};

\draw (3.4,3) -- (3.4, 3.9) -- (3.9, 3.9) -- (3.9, 3);
\draw (3.4, 5) -- (3.4, 4.1) -- (3.9, 4.1) -- (3.9, 5);

\node at (1,1) {$\int dO$};

\draw (1.7,0) -- (1.7, 2);
\draw [fill = white] (1.7,1) circle [radius=0.2];
\draw (2.2,0) -- (2.2, 2);
\draw [fill = white] (2.2,1) circle [radius=0.2];
\draw (2.7,0) -- (2.7, 2);
\draw [fill = white] (2.7,1) circle [radius=0.2];
\draw (3.2,0) -- (3.2, 2);
\draw [fill = white] (3.2,1) circle [radius=0.2];

\node at (4, 1) {$ = \frac{1}{N^2} $};

\node at (4.6, 1) {$\Bigg\{$};

\draw (5,0) -- (5,0.9) -- (6.5, 0.9) -- (6.5, 0);
\draw (5,2) -- (5, 1.1) -- (6.5, 1.1) -- (6.5, 2);
\draw (5.5, 0) -- (5.5, 0.7) -- (6, 0.7) -- (6, 0);
\draw (5.5, 2) -- (5.5, 1.3) -- (6, 1.3) -- (6, 2);

\node at (6.8, 1) {$+$};

\draw (7.1, 0) -- (7.1, 0.9) -- (7.6, 0.9) -- (7.6, 0);
\draw (7.1, 2) -- (7.1, 1.1) -- (7.6, 1.1) -- (7.6, 2);
\draw (8.1, 0) -- (8.1, 0.9) -- (8.6, 0.9) -- (8.6, 0);
\draw (8.1, 2) -- (8.1, 1.1) -- (8.6, 1.1) -- (8.6, 2);

\node at (8.9, 1) { $ + $};

\draw (9.2, 0) -- (9.2, 0.9) -- (10.2, 0.9) -- (10.2, 0);
\draw (9.2, 2) -- (9.2, 1.1) -- (10.2, 1.1) -- (10.2, 2);
\draw (9.7, 0) -- (9.7, 0.7) --  (10, 0.7) ;
\draw (10.4, 0.7) -- (10.6, 0.7) -- (10.6, 0);
\draw (9.7, 2) -- (9.7, 1.3) -- (10, 1.3);
\draw (10.4, 1.3) -- (10.6, 1.3) -- (10.6, 2);

\node at (11, 1) {$\Bigg\}$};

\node at (11.9, 1) {$ + \,  o(\frac{1}{N^2})$};

\end{tikzpicture}
\caption{\textit{The graphical representation of the cases $k =1$ and $k=2$ of the asymptotic integration formula above. Note that in the $k=1$ case, the dominant asymptotic is the exact answer.}}
\label{fig2}
\end{figure}

\subsection*{The Mean and the Variance}

We are now ready to compute the asymptotic mean and variance of the measure $d\nu_N[t]$. Let us begin with the mean $t_0$:

\besn
t_0 & = & \int t d\nu_N[t] = \int \mathcal{D} O t(O) =  \frac{1}{2N V} \int \mathcal{D} O \, \tr \Big (  \ln (K_{aa'}   \Big ) \\
& = & \frac{1}{2NV} \int \mathcal{D} O \, \tr \Big ( \ln ( (- \partial^2 + \mu) \delta_{aa'} + O^t \hat{M}_{aa'} O    )   \Big ) \\
& = & \frac{1}{2V} \tr \ln ( - \partial^2 + \mu) + \frac{1}{2NV} \int \mathcal{D} O \tr \ln \Big ( \delta_{aa'} + (- \partial^2 + \mu)^{-1} O^t \hat{M}_{aa'}O)   \Big ).
\eesn

Now, before proceeding any further, we shall make the simplifying ansatz of taking the field $M_{aa'}$ to be independent of space. The justification for this is that we anticipate, similarly to what happens in the large-$N$ vector model, that the configuration of this field that gives the dominant asymptotic of the partition function is translation invariant. \newline

Keeping this simplification in mind, we get\footnote{Of course, as it stands, this is only justified rigorously provided that $||(- \partial^2 + \mu)^{-1} O^t \hat{M}O  || < 1$, which in turn can be guaranteed if one assumes that $ || (- \partial^2 + \mu)^{-1} || || O^t \hat{M}O || =  \frac{||  O^t \hat{M}O  ||}{\mu}  = \frac{||\hat{M}||}{\mu}< 1 $. However, one can justify the equality of this integral with its evaluation (derived below) on a much bigger domain by analytic continuation. Needless to say, one would need to check whether the dominant contribution to the asymptotic will come from a point inside the domain of analyticity. Since, for reasons that will be apparent at the end of the paper, we will not need to calculate the distribution of the $\hat{M}$ giving the dominant asymptotic, this more careful analysis will be performed elsewhere.}

\besn
\frac{1}{2NV} \int \mathcal{D} O \tr \ln \Big ( \delta_{aa'} + (- \partial^2 + \mu)^{-1} O^t \hat{M}O   \Big ) & = & \\
\frac{1}{2 N V} \sum_{n=1}^\infty \int \mathcal{D} O \tr \frac{(-1)^n}{n}  \Big ( (- \partial^2 + \mu)^{-1} O^t \hat{M}O   \Big )^n.
\eesn

Focusing now on a particular term in the sum over $n$ and reverting to the discrete notation for clarity, we get that

\besn
\frac{(-1)^n}{2nNV}\int \mathcal{D} O \tr   \Big ( (- \partial^2 + \mu)^{-1} O^t \hat{M}O   \Big )^n = \\
 \frac{(-1)^n}{2nNV} \sum_{x_1, \dots, x_n} \sum_{a_1, b_1, c_1, \dots, a_n, b_n, c_n}  \int \mathcal{D}O (- \partial^2 + \mu)^{-1}_{x_1, x_2} O^t_{a_1 b_1}(x_2) \hat{M}_{b_1 c_1} O_{c_1 a_2} (x_2) \dots    \\
\dots (- \partial^2 + \mu)^{-1}_{x_n, x_1} O^t_{a_n b_n} (x_1) \hat{M}_{b_n c_n} O_{c_n a_1} (x_1) .
\eesn

The sum over the $x$'s can be split into the following classes: the class where all the $x$'s are different, the class where all are different except two, the class where all different except three and so on. We shall consider only the first two classes since, as will be apparent below, the rest are suppressed.\newline

Let us now perform the $\mathcal{D} O$ integration in these two cases. In the case where all the $x$'s are different, we get $n$ repetitions of the following graphical formula

\begin{figure}[h]
\begin{tikzpicture}

\node at (3, 1) {$\int d O$};

\draw (4, 1.5) -- (4.5, 2) -- (5, 1.5) -- (4,1.5 );
\draw (4.2, 1.5) -- (4.2, 0.5) ;
\draw [dashed] (4.2, 0.5) -- (3.8, 0.5);
\draw (4.8, 1.5) -- (4.8, 0.5);
\draw [dashed] (4.8, 0.5) -- (5.2, 0.5);
\draw [fill = white] (4.2,1) circle [radius=0.2];
\draw [fill = white] (4.8,1) circle [radius=0.2];

\node at (6, 1) {$ =  \, \frac{1}{N}$};

\draw (7, 0.5) -- (7, 0.9) -- (7.5, 0.9) -- (7.5, 0.5);
\draw (7, 1.5) -- (7, 1.1) -- (7.5, 1.1) -- (7.5, 1.5);
\draw [dashed] (7, 0.5) -- (6.6, 0.5);
\draw [dashed] (7.5, 0.5) -- (7.9, 0.5);
\draw (6.75, 1.5) -- (7.25, 2) -- (7.75, 1.5) -- (6.75, 1.5);

\end{tikzpicture}

\caption{\textit{ The dashed lines represent indices which are contracted with other terms. The triangle stands for the matrix $\hat{M}_{aa'}$.   }}
\label{single}

\end{figure}

Performing all the integrals, we will be left with the following expression

\begin{figure}[h]
\begin{tikzpicture}[scale = 0.97]

\node at (1.3,1.5) {$\Bigg ( \, \frac{1}{N} $};
\draw (2, 1.5) -- (2, 1.1) -- (2.5, 1.1) -- (2.5, 1.5);
\draw (1.75, 1.5) -- (2.25, 2) -- (2.75, 1.5) -- (1.75, 1.5);
\node at (3, 1.5) {$\Bigg )^n$};
\draw (3.2,1.1) -- (3.2, 1.8) -- (3.9,1.8) -- (3.9,1.1) -- (3.2,1.1);

\node at (8.8, 1.5) {$ =  \Big ( \frac{1}{N} \sum_a \hat{M}_{a} \Big )^n \Big ( \sum_b \delta_{bb} \Big )  = N \Big ( \int d \hat{M} \rho(\hat{M}) \hat{M} \Big )^n \equiv N \Big ( \overline{\hat{M}}  \Big )^n$  };

\end{tikzpicture}

\caption{\textit{ The result of performing the $O$ integral in the case of non-coincident points. The bar denotes averaging with respect to the $\rho$ measure. }}

\end{figure}

In the case when all the $x$'s are different except two, the $O$ integrations over those $O$'s located at the non-coincident points proceeds as above. The graphical calculation of integration over the remaining $O$'s is shown in figure \ref{coincident}. \newline

\begin{figure}[h!]
\begin{tikzpicture}[scale = 1]

\node at (0.3,8.5) {$\Bigg ( \, \frac{1}{N} $};
\draw (1, 8.5) -- (1, 8.1) -- (1.5, 8.1) -- (1.5, 8.5);
\draw (0.75, 8.5) -- (1.25, 9) -- (1.75, 8.5) -- (0.75, 8.5);
\node at (2.2, 8.5) {$\Bigg )^{n-2}$};

\node at (3, 8.5) {$\int d O$};

\draw (3.5, 8.7) -- (4, 9.2) -- (4.5, 8.7) -- (3.5,8.7);
\draw (3.7, 8.7) -- (3.7, 7.4) ;
\draw (4.3, 8.7) -- (4.3, 7.7);
\draw (4.3, 7.7) -- (5, 7.7);
\draw [fill = white] (3.7,8.2) circle [radius=0.2];
\draw [fill = white] (4.3,8.2) circle [radius=0.2];

\draw (4.8, 8.7) -- (5.3, 9.2) -- (5.8, 8.7) -- (4.8,8.7 );
\draw (5, 8.7) -- (5, 7.7) ;
\draw (5.6, 8.7) -- (5.6, 7.4);
\draw [fill = white] (5,8.2) circle [radius=0.2];
\draw [fill = white] (5.6,8.2) circle [radius=0.2];
\draw (5.6, 7.4) -- (3.7, 7.4);

\node at (6.2, 8.5) {$=$};

\node at (0.3,5.5) {$\Bigg ( \, \frac{1}{N} $};
\draw (1, 5.5) -- (1, 5.1) -- (1.5, 5.1) -- (1.5, 5.5);
\draw (0.75, 5.5) -- (1.25, 6) -- (1.75, 5.5) -- (0.75, 5.5);
\node at (2.2, 5.5) {$\Bigg )^{n-2}$};

\node at (3, 5.5) {$\int d O$};

\draw (3.5, 5.7) -- (4, 6.2) -- (4.5, 5.7) -- (3.5,5.7);
\draw (3.7, 5.7) -- (3.7, 4.7) ;
\draw (4.3, 5.7) -- (4.3, 5.3);
\draw (4.3, 5.3) -- (5, 5.3);
\draw [fill = white] (3.7,5.2) circle [radius=0.2];

\draw (4.8, 5.7) -- (5.3, 6.2) -- (5.8, 5.7) -- (4.8,5.7 );
\draw (5, 5.7) -- (5, 5.3) ;
\draw (5.6, 5.7) -- (5.6, 4.7);
\draw [fill = white] (5.6,5.2) circle [radius=0.2];
\draw (5.6, 4.7) -- (3.7, 4.7);

\node at (6.2, 5.5) {$= $};

\node at (6.9,5.5) {$\Bigg ( \, \frac{1}{N} $};
\draw (7.6, 5.5) -- (7.6, 5.1) -- (8.1, 5.1) -- (8.1, 5.5);
\draw (7.35, 5.5) -- (7.85, 6) -- (8.35, 5.5) -- (7.35, 5.5);
\node at (8.8, 5.5) {$\Bigg )^{n-2}$};

\draw (9.2, 5.7) -- (9.7, 6.2) -- (10.2, 5.7) -- (9.2,5.7);
\draw (9.4, 5.7) -- (9.4, 4.7) ;
\draw (10, 5.7) -- (10, 5.3);
\draw (10, 5.3) -- (10.7, 5.3);

\draw (10.5, 5.7) -- (11, 6.2) -- (11.5, 5.7) -- (10.5,5.7 );
\draw (10.7, 5.7) -- (10.7, 5.3) ;
\draw (11.3, 5.7) -- (11.3, 4.7);

\draw (11.3, 4.7) -- (9.4, 4.7);

\node at (11.7 , 5.5) {$=$};

\node at (  6.1 , 3   ) {$ \Big ( \frac{1}{N} \sum_a \hat{M}_{a} \Big )^{n-2}     \Big ( \sum_b \hat{M}_b^2  \Big )   = $ };
\node at ( 6.1, 1.5) { $N \Big ( \int d \hat{M} \rho(\hat{M}) \hat{M} \Big )^{n-2} \Big ( \int d \hat{M} \rho(\hat{M}) \hat{M}^2  \Big )  \equiv  N  \Big ( \overline{\hat{M}}  \Big )^{n-2} \Big ( \overline{\hat{M}^2} \Big )$};

\end{tikzpicture}

\caption{\textit{ The result of doing the $O$ integrals when exactly two $x$'s are coincident. As above, the bar denotes averaging with respect to the $\rho$ measure. The starting expression is that obtained after integrating over the $O$'s at non-coincident points. The first and second equalities above follow from the fact that $O^t O = I$ and that the $dO$ is a probability measure.   }}
\label{coincident}
\end{figure}

In order to proceed, we use the following, easily verified formula

 \ben
 (- \partial^2 + \mu)^{-1}_{x,y} = \frac{1}{\Lambda^2} \sum_p  \frac{e^{i p (x-y)}}{p^2 + \mu}.
 \een

Putting it all together, we have that 

\besn
\frac{(-1)^n}{2nNV} \sum_{x_1, \dots, x_n} \sum_{a_1, b_1, c_1, \dots, a_n, b_n, c_n}  \int \mathcal{D}O (- \partial^2 + \mu)_{x_1, x_2} O^t_{a_1 b_1}(x_2) \hat{M}_{b_1 c_1} O_{c_1 a_2} (x_2) \dots   \\
\dots (- \partial^2 + \mu)_{x_n, x_1} O^t_{a_n b_n} (x_1) \hat{M}_{b_n c_n} O_{c_n a_1} (x_1) \\
= \frac{(-1)^n \Big (\overline{\hat{M}}\Big )^n }{2nV}  \Bigg (    \sum_{x_1, \dots, x_n} (- \partial^2 + \mu)^{-1}_{x_1, x_2} \dots (- \partial^2 + \mu)^{-1}_{x_n, x_1}   \\
  +  {n \choose 2}\bigg ( \frac{\overline{\hat{M}^2}}{\overline{\hat{M}}^2}  -1  \bigg )  \sum_{x_1, \dots, x_{n-1}}  (- \partial^2 + \mu)^{-1}_{x_1, x_2} \dots (- \partial^2 + \mu)^{-1}_{x_{n-1}, x_1} + \dots \\
= \frac{1}{(2 \pi)^2} \Bigg [  \frac{  (-1)^n \Big ( \overline{ \hat{M}} \Big )^n   }{2 n}  + \frac{(-1)^n (n-1) \Big ( \overline{\hat{M}}  \Big )^n}{4 \Lambda^2} \bigg ( \frac{\overline{\hat{M}^2}}{\overline{\hat{M}}^2}  -1  \bigg ) + o \Big ( \frac{1}{\Lambda^2} \Big )  \Bigg ] \int \frac{d^2p}{ (p^2 + \mu )^n}.
\eesn

Therefore,

\besn
t_0 & = & \frac{1}{2V} \tr \ln (- \partial^2 + \mu) + \frac{1}{2}  \int d^2p \ln \Big ( 1 + \frac{ \overline{\hat{M}}}{p^2+\mu}   \Big ) + O(1) \\
& = & \frac{1}{2 (2 \pi)^2} \int d^2 p \ln (p^2 + \mu + \overline{\hat{M}}) + O(1).
\eesn

Note that the correction is indeed of order 1 since the second term in the square brackets above when summed over $n$ gives an integrand of order 1. The latter, when integrated, gives a term of order $\Lambda^2$ which in turn cancels with the identical factor in the denominator. \newline

Before we compute the variance, let us compute $\int \mathcal{D} O  \int J_{ab} (K_{aa'} )^{-1} J_{a'b}$, since the relevant calculation is almost identical with the one done above to obtain $t_0$. In fact, we have

\besn
& & \int \mathcal{D} O  \int J_{ab} (K_{aa'} )^{-1} J_{a'b}  
= \int \int \mathcal{D} O J_{ab} \Big (- \partial^2 + \mu + O^t \hat{M} O  \Big  )^{-1} _{aa'}    J_{a'b}  \\
 & & =\int \int \mathcal{D} O J_{ab} ( - \partial^2 + \mu   )^{-1} \Big  ( I + (- \partial^2 + \mu)^{-1}O^t \hat{M} O    \Big)^{-1}_{aa'} J_{a'b} \\
 & & = \sum_{n=0}^\infty (-1)^n  \int \int \mathcal{D} O J_{ab} (- \partial^2 + \mu)^{-1}   \Big (   (- \partial^2 + \mu)^{-1}O^t \hat{M} O    \Big)_{aa'}^n     J_{a'b}.
\eesn

The calculation then proceeds by integrating over the $O$'s as shown in figure \ref{single}. Thus,

\ben
\int \mathcal{D} O  \int J_{ab} (K_{aa'} )^{-1} J_{a'b} = \int J_{ab} \Big (- \partial^2 + \mu + \overline{\hat{M}} \Big )^{-1} J_{ab}.
\een

Moving on to computing the variance of $d \nu_N[t]$, we need to compute 

\besn
A^{-1}  &  = & \int (t^2 - t_0^2 ) d \nu_N[t] = \int \mathcal{D} O t^2(O) - \Big ( \int \mathcal{D} O t(O) \Big )^2 \\
& = & \frac{1}{4 N^2 V^2}  \bigg ( \int \mathcal{D} O \Big ( \tr \ln \big (  \delta_{aa'} + (- \partial^2 + \mu)^{-1} O^t \hat{M}_{aa'}O) \big ) \Big )^2  \\
 & & - \Big ( \int \mathcal{D} O \tr \ln ( \delta_{aa'} + (- \partial^2 + \mu)^{-1} O^t \hat{M}_{aa'}O)   )      \Big )^2          \bigg )  \\
 & = & \frac{1}{4 N^2 V^2} \sum_{n,m} \frac{(-1)^{n+m}}{nm}  \bigg [ \int \mathcal{D} O \tr  \Big ( (- \partial^2 + \mu)^{-1} O^t \hat{M}O   \Big )^n   \\
 & & \times \tr  \Big ( (- \partial^2 + \mu)^{-1} O^t \hat{M}O   \Big )^m  -  \Big ( \int \mathcal{D} O \tr  \Big ( (- \partial^2 + \mu)^{-1} O^t \hat{M}O   \Big )^n   \\
 & &  \times \int \mathcal{D}O'  \tr \Big ( (- \partial^2 + \mu)^{-1} O'^t \hat{M}O'   \Big )^m  \bigg ].
\eesn

Consider the term in the square brackets

\besn
\int \mathcal{D} O \tr  \Big ( (- \partial^2 + \mu)^{-1} O^t \hat{M}O   \Big )^n  \tr  \Big ( (- \partial^2 + \mu)^{-1} O^t \hat{M}O   \Big )^m  - &  & \\
 -    \int \mathcal{D} O \tr  \Big ( (- \partial^2 + \mu)^{-1} O^t \hat{M}O   \Big )^n   \int \mathcal{D}O'  \tr \Big ( (- \partial^2 + \mu)^{-1} O'^t \hat{M}O'   \Big )^m  &   & \\
 = \sum_{x_1, \dots, x_n; y_1, \dots, y_m} (- \partial^2 + \mu)^{-1}_{x_1, x_2} \dots (- \partial^2 + \mu)^{-1}_{x_n, x_1}  (- \partial^2 + \mu)^{-1}_{y_1, y_2} & \dots & \\
 \dots  (- \partial^2 + \mu)^{-1}_{y_m, y_1} 
    \bigg [ \int \mathcal{D} O \txt{tr} \Big ( O^t(x_1) \hat{M} O(x_1) \dots O^t(x_n) \hat{M} O(x_n) \Big )  &  & \\
 \times \txt{tr} \Big ( O^t(y_1) \hat{M} O (y_1) \dots O^t(y_m) \hat{M} O(y_m) \Big )  &  & \\
  - \int \mathcal{D} O \, \txt{tr} \Big ( O^t(x_1) \hat{M} O( x_1) \dots O^t(x_n) \hat{M} O(x_n) \Big )  &  & \\ 
  \times \int \mathcal{D} O' \, \txt{tr} \Big ( O'^t(y_1) \hat{M} O'(y_1) \dots O'^t(y_m) \hat{M} O'(y_m) \Big ) \bigg ] .
\eesn

The $\mathcal{D}O$ and $\mathcal{D} O'$ integrals will depend, as before, on the fact whether there are any coincidences between the points at which the relevant $O$'s and $O'$'s are located. It is obvious that unless there are coincidences between the $x$'s and the $y$'s, the terms in the square brackets cancel each other. Therefore, we only need to consider the situation with coincidences between the $x$'s and the $y$'s. Of course, each extra coincidence ``costs'' a factor of $\frac{1}{\Lambda^2}$, and thus, we need only look at the minimal possible number of coincidences which give nonzero answers. It is easy to see that if there is a single $x$ coinciding with a single $y$, then the two terms still cancel out. The reason for this is that after doing  the integral over all the other $O$'s and $O'$'s, the dependence on the remaining orthogonal matrix will disappear (similarly to how it happened in the calculation in figure \ref{coincident}). We thus need to consider the case when there are precisely two pairs of coincident points. Performing the integrals over all the other points, we see that the term in brackets is calculated graphically as shown in figure \ref{variance}.

\begin{figure}[h]
\begin{tikzpicture}

\node at (2, 11) {$\Bigg ($};

\node at (2.5, 11) {$\int d$};
\draw [fill = white] (3,11) circle [radius  = 0.2];
\node at (3.4, 11) {$ d$};
\draw [fill = gray] (3.8,11) circle [radius  = 0.2];

\draw (4.3, 11.5) -- (4.8, 12) -- (5.3, 11.5) -- (4.3,11.5 );
\draw (4.5, 11.5) -- (4.5, 10.5) ;
\draw (5.1, 11.5) -- (5.1, 10.5);
\draw [fill = white] (4.5,11) circle [radius=0.2];
\draw [fill = white] (5.1,11) circle [radius=0.2];

\draw (5.5, 11.5) -- (6, 12) -- (6.5, 11.5) -- (5.5,11.5 );
\draw (5.7, 11.5) -- (5.7, 10.5) ;
\draw (6.3, 11.5) -- (6.3, 10.5);
\draw [fill = white] (5.7,11) circle [radius=0.2];
\draw [fill = white] (6.3,11) circle [radius=0.2];

\draw (6.7, 11.5) -- (7.2, 12) -- (7.7, 11.5) -- (6.7,11.5 );
\draw (6.9, 11.5) -- (6.9, 10.5) ;
\draw (7.5, 11.5) -- (7.5, 10.5);
\draw [fill = gray] (6.9,11) circle [radius=0.2];
\draw [fill = gray] (7.5,11) circle [radius=0.2];

\draw (7.9, 11.5) -- (8.4, 12) -- (8.9, 11.5) -- (7.9,11.5 );
\draw (8.1, 11.5) -- (8.1, 10.5) ;
\draw (8.7, 11.5) -- (8.7, 10.5);
\draw [fill = gray] (8.1,11) circle [radius=0.2];
\draw [fill = gray] (8.7,11) circle [radius=0.2];

\draw (6.3, 10.5) -- (6.9, 10.5);
\draw (5.7, 10.5) -- (5.7, 10.3) -- (7.5, 10.3) -- (7.5, 10.5);
\draw (5.1, 10.5) -- (5.1, 10.1) -- (8.1, 10.1) -- (8.1, 10.5);
\draw (4.5, 10.5) -- (4.5, 9.9) -- (8.7, 9.9) -- (8.7, 10.5);

\node at (11, 11) { $-   N^2  \big (\overline{\hat{M}} \big )^4 \Bigg ) \big ( \overline{\hat{M}} \big ) ^{ n+ m - 4}  $  };

\node at (3.5, 7.5) {$ = \frac{1}{N^4} \Bigg \{ $};

\node at (4.4, 7.5) {$2 \times$};

\draw (4.8, 8.5) -- (5.3, 9) -- (5.8, 8.5) -- (4.8,8.5 );
\draw (5, 8.5) -- (5, 8.2) ;
\draw (5.6, 8.5) -- (5.6, 8.2);

\draw (6, 8.5) -- (6.5, 9) -- (7, 8.5) -- (6,8.5 );
\draw (6.2, 8.5) -- (6.2, 8.2) ;
\draw (6.8, 8.5) -- (6.8, 8.2);

\draw (7.2, 8.5) -- (7.7, 9) -- (8.2, 8.5) -- (7.2,8.5 );
\draw (7.4, 8.5) -- (7.4, 8.2) ;
\draw (8, 8.5) -- (8, 8.2);

\draw (8.4, 8.5) -- (8.9, 9) -- (9.4, 8.5) -- (8.4,8.5 );
\draw (8.6, 8.5) -- (8.6, 8.2) ;
\draw (9.2, 8.5) -- (9.2, 8.2);

\draw (6.8, 7.5) -- (7.4, 7.5);
\draw (6.2, 7.5) -- (6.2, 7.3) -- (8, 7.3) -- (8, 7.5);
\draw (5.6, 7.5) -- (5.6, 7.1) -- (8.6, 7.1) -- (8.6, 7.5);
\draw (5, 7.5) -- (5, 6.9) -- (9.2, 6.9) -- (9.2, 7.5);

\draw (5.6, 8.2) -- (6.2, 8.2);
\draw (5, 8.2) -- (5, 8) -- (6.8, 8) -- (6.8, 8.2);
\draw (7.4, 8.2) -- (8, 8.2);
\draw (8.6, 8.2) -- (9.2, 8.2);

\draw (7.4, 7.8) -- (8, 7.8);
\draw (8.6, 7.8) -- (9.2, 7.8);

\draw (5.6, 7.5) -- (6.2, 7.5);
\draw (5, 7.5) --(5, 7.8) -- (6.8, 7.8) -- (6.8, 7.5);
\draw (7.4, 7.5) -- (7.4, 7.8);
\draw (8, 7.5) -- (8, 7.8);
\draw (8.6, 7.5) -- (8.6, 7.8);
\draw (9.2, 7.5) -- (9.2, 7.8);

\node at (12, 7.5) {$+$};

\node at (4.4, 4.5) {$2 \times$};

\draw (4.8, 5.5) -- (5.3, 6) -- (5.8, 5.5) -- (4.8,5.5 );
\draw (5, 5.5) -- (5, 5.2) ;
\draw (5.6, 5.5) -- (5.6, 5.2);

\draw (6, 5.5) -- (6.5, 6) -- (7, 5.5) -- (6,5.5 );
\draw (6.2, 5.5) -- (6.2, 5.2) ;
\draw (6.8, 5.5) -- (6.8, 5.2);

\draw (7.2, 5.5) -- (7.7, 6) -- (8.2, 5.5) -- (7.2,5.5 );
\draw (7.4, 5.5) -- (7.4, 5.2) ;
\draw (8, 5.5) -- (8, 5.2);

\draw (8.4, 5.5) -- (8.9, 6) -- (9.4, 5.5) -- (8.4,5.5 );
\draw (8.6, 5.5) -- (8.6, 5.2) ;
\draw (9.2, 5.5) -- (9.2, 5.2);

\draw (6.8, 4.5) -- (7.4, 4.5);
\draw (6.2, 4.5) -- (6.2, 4.3) -- (8, 4.3) -- (8, 4.5);
\draw (5.6, 4.5) -- (5.6, 4.1) -- (8.6, 4.1) -- (8.6, 4.5);
\draw (5, 4.5) -- (5, 3.9) -- (9.2, 3.9) -- (9.2, 4.5);

\draw (5, 8.2) -- (5, 8) -- (6.8, 8) -- (6.8, 8.2);
\draw (7.4, 5.2) -- (8, 5.2);
\draw (8.6, 5.2) -- (9.2, 5.2);

\draw (7.4, 4.8) -- (8, 4.8);
\draw (8.6, 4.8) -- (9.2, 4.8);

\draw (5.6, 7.5) -- (6.2, 7.5);
\draw (5, 7.5) --(5, 7.8) -- (6.8, 7.8) -- (6.8, 7.5);
\draw (7.4, 4.5) -- (7.4, 4.8);
\draw (8, 4.5) -- (8, 4.8);
\draw (8.6, 4.5) -- (8.6, 4.8);
\draw (9.2, 4.5) -- (9.2, 4.8);

\draw (5.6, 5.2) -- (5.6, 5) -- (6.8, 5) -- (6.8, 5.2);
\draw (5, 5.2) -- (5.5, 5.2);
\draw (5.7, 5.2) -- (6.2, 5.2);
\draw (5.6, 4.5) --(5.6, 4.8) -- (6.8, 4.8) -- (6.8, 4.5);
\draw (5, 4.5) -- (5.5, 4.5);
\draw (5.7, 4.5) -- (6.2, 4.5);

\node at (12, 4.5) {$+$};

\draw (4.8, 2.5) -- (5.3, 3) -- (5.8, 2.5) -- (4.8,2.5 );
\draw (5, 2.5) -- (5, 2.2) ;
\draw (5.6, 2.5) -- (5.6, 2.2);

\draw (6, 2.5) -- (6.5, 3) -- (7, 2.5) -- (6,2.5 );
\draw (6.2, 2.5) -- (6.2, 2.2) ;
\draw (6.8, 2.5) -- (6.8, 2.2);

\draw (7.2, 2.5) -- (7.7, 3) -- (8.2, 2.5) -- (7.2,2.5 );
\draw (7.4, 2.5) -- (7.4, 2.2) ;
\draw (8, 2.5) -- (8, 2.2);

\draw (8.4, 2.5) -- (8.9, 3) -- (9.4, 2.5) -- (8.4,2.5 );
\draw (8.6, 2.5) -- (8.6, 2.2) ;
\draw (9.2, 2.5) -- (9.2, 2.2);

\draw (6.8, 1.5) -- (7.4, 1.5);
\draw (6.2, 1.5) -- (6.2, 1.3) -- (8, 1.3) -- (8, 1.5);
\draw (5.6, 1.5) -- (5.6, 1.1) -- (8.6, 1.1) -- (8.6, 1.5);
\draw (5, 1.5) -- (5, 0.9) -- (9.2, 0.9) -- (9.2, 1.5);

\draw (5.6, 2.2) -- (6.2, 2.2);
\draw (5, 2.2) -- (5, 2) -- (6.8, 2) -- (6.8, 2.2);
\draw (8, 2.2) -- (8.6, 2.2);
\draw (7.4, 2.2) -- (7.4, 2) -- (9.2, 2) -- (9.2, 2.2);

\draw (5.6, 1.5) -- (6.2, 1.5);
\draw (5, 1.5) --(5, 1.8) -- (6.8, 1.8) -- (6.8, 1.5);

\draw (8, 1.5) -- (8.6, 1.5);
\draw (7.4, 1.5) --(7.4, 1.8) -- (9.2, 1.8) -- (9.2, 1.5);

\node at (11.5, 1.5) {$ \quad +  \quad o (N^4) \Bigg \} \big ( \overline{\hat{M}} \big ) ^{ n+ m - 4}$};

\node at (7, 0) {$=  \Big (  4  \big ( \overline{\hat{M}} \big )^2 \, \overline{\hat{M}^2}  \quad + \quad  \big ( \overline{\hat{M}^2}  \big )^2 \quad  + \quad  o(1)  \Big ) \big ( \overline{\hat{M}} \big ) ^{ n+ m - 4} $};

\end{tikzpicture}

\caption{\textit{ The starting expression above is that obtained after integration over all the $O$'s at non-coincident points. The $o(N^4)$ terms in the curly braces come from the remaining two ways of reconnecting the edges.   }}
\label{variance}

\end{figure}

Thus, we have that, up to subdominant terms,

\besn
A^{-1} = \frac{1}{4 N^2 V^2} \sum_{n,m} \frac{  (-1)^{n+m} } {nm}  \Big (  4  \big ( \overline{\hat{M}} \big )^2 \, \overline{\hat{M}^2}  \quad + \quad  \big ( \overline{\hat{M}^2}  \big )^2 \quad   \Big ) \big ( \overline{\hat{M}} \big ) ^{ n+ m - 4}    &  & \\  
 \times \sum_{x_1, \dots, x_n; y_1, \dots, y_m; x_i = y_{i'}, x_j = y_{j'}} (- \partial^2 + \mu)^{-1}_{x_1, x_2} 
\dots (- \partial^2 + \mu)^{-1}_{x_n, x_1} \times & \dots & \\
 \dots \times (- \partial^2 + \mu)^{-1}_{y_1, y_2} \dots (- \partial^2 + \mu)^{-1}_{y_m, y_1} &  & \\
= \frac{  4 \frac{  \overline{\hat{M}^2}    }{  ( \overline{\hat{M}}  ) ^2   } + \frac{ ( \overline{\hat{M}^2}   )^2    }{  ( \overline{\hat{M}}  ) ^4   }      }{4 N^2 V^2 \Lambda^ 8} \sum_{n, m} (-1)^{n+m}  ( \overline{\hat{M}}  ) ^{ n+ m }     \sum_{x,y} \sum_{\alpha = 1}^{n-1} \sum_{\beta = 1}^{m-1} \sum_{p, p', q, q'} e^{i (p + p' + q + q') (x-y)}  &  & \\
 \times \frac{1}{(p^2 + \mu)^\alpha}\frac{1}{( p'^2 + \mu)^{n- \alpha}} \frac{1}{(q^2 + \mu)^\beta} \frac{1}{(q'^2 + \mu)^{m - \beta}} &  &
\eesn
\besn
 = \frac{  4 \frac{  \overline{\hat{M}^2}    }{  ( \overline{\hat{M}}  ) ^2   } + \frac{ ( \overline{\hat{M}^2}   )^2    }{  ( \overline{\hat{M}}  ) ^4   }      }{4 N^2 V^2 \Lambda^ 8} \sum_{n, m} (-1)^{n+m}  ( \overline{\hat{M}}  ) ^{ n+ m }     \sum_{x,y}  \sum_{p, p', q, q'} e^{i (p + p' + q + q') (x-y)}  &  & \\
 \times \bigg ( \frac{1}{p^2 + \mu}\frac{1}{p'^2 + \mu} \frac{( \frac{1}{p^2 + \mu})^{n-1} - ( \frac{1}{p'^2 + \mu}) ^{n-1}}{\frac{1}{p^2 + \mu} - \frac{1}{p'^2 + \mu}}   \bigg ) \\
 \times  \bigg ( \frac{1}{q^2 + \mu}\frac{1}{q'^2 + \mu} \frac{( \frac{1}{q^2 + \mu})^{n-1} - ( \frac{1}{q'^2 + \mu}) ^{n-1}}{\frac{1}{q^2 + \mu} - \frac{1}{q'^2 + \mu}}  \bigg ) &  & \\
 = \frac{  4  \overline{\hat{M}^2} ( \overline{\hat{M}}  ) ^4    +  ( \overline{\hat{M}^2}   )^2   }{4 N^2 V^2 \Lambda^ 8}    \Bigg [ \sum_{x,y}  \sum_{p, p', q, q'} e^{i (p + p' + q + q') (x-y)}  &  & \\
 \times \bigg ( \frac{1}{p^2 + \mu}\frac{1}{p'^2 + \mu}  \frac{1}{q^2 + \mu}\frac{1}{q'^2 + \mu}     \bigg ) \bigg \{ \frac{p^2 + \mu}{p^2 + \mu + \overline{\hat{M}}} -  \frac{p'^2 + \mu}{p'^2 + \mu + \overline{\hat{M}}}  \bigg \}  &  &  \\
  \times \bigg \{ \frac{q^2 + \mu}{q^2 + \mu + \overline{\hat{M}}} -  \frac{q'^2 + \mu}{q'^2 + \mu + \overline{\hat{M}}}  \bigg \} \Bigg ].
\eesn

We now obtain the continuum limit asymptotics of the term in the square brackets above. In this limit, note that the terms in the curly brackets are asymptotic to 1. We thus consider

\besn
\sum_{x,y}  \sum_{p, p', q, q'} e^{i (p + p' + q + q') (x-y)} 
 \bigg ( \frac{1}{p^2 + \mu}\frac{1}{p'^2 + \mu}  \frac{1}{q^2 + \mu}\frac{1}{q'^2 + \mu}     \bigg )   & & \\
 \sim \frac{V^4}{( 2 \pi)^8   } \sum_{x,y} \bigg ( \int \frac{e^{i p(x-y)}   }{p^2 + \mu} dp \bigg )^4  \sim  \frac{ \Lambda^4 V^2}{ (2 \pi)^8} \int dx dy \bigg ( \int \frac{e^{i p(x-y)}   }{p^2 + \mu} dp \bigg )^4  \simeq \Lambda^4 V^3,   & & 
\eesn

where $\simeq$ stands for ``asymptotic up to a constant''. Putting everything together, we see that 

\ben
A^{-1} \simeq \frac{1}{N^2} \Big ( 4  \overline{\hat{M}^2} ( \overline{\hat{M}}  ) ^4    +  ( \overline{\hat{M}^2}   )^2 \Big ) \frac{V}{\Lambda^4}. 
\een

\vspace{0.5cm}

\subsection*{Large-$N$ Asymptotics}

We are now ready to use everything we've learned and plug it in (\ref{eqnarray:asymp}). Then, we get 

\bes
\label{eqnarray:asympto}
Z[J] \sim \int \mathcal{D} \rho dt \exp \bigg [ - \frac{1}{2} \int J_{ab} \Big (- \partial^2 + \mu + \overline{\hat{M}} \Big )^{-1} J_{ab}    \bigg ]   & & \nonumber \\
 \times   \Bigg ( \exp \bigg [  N^2 \bigg ( \frac{\Lambda^2}{2} \int d\hat{M} d\hat{M}'  \rho(\hat{M}) \rho(\hat{M}') \ln | \hat{M} - \hat{M}' | + \frac{V \mu}{2 \lambda}   +  \frac{\overline{\hat{M}}}{2 \lambda}  \bigg )   \bigg ] \Bigg )  \nonumber & & \\
  \times \Bigg (  e^{ - N^2 ( V t + \frac{A}{2} (t-t_0)^2    )   }   \Bigg ) . 
\ees

This expression is now in the appropriate form for the large-$N$ limit analysis using the usual Laplace's method. Before we do so, let us briefly sketch an alternative way, alluded to after the lemma above, to arrive at the expression we have. \newline

Instead of the argument given above, using the Lipschitzness of the $J$ term which allowed us to take the aforementioned term outside the $\mathcal{D}O$ integral, we could have proceeded as above, pushing forward the measure directly, with the change that 

\ben
\tilde{t}(O) = \frac{1}{2NV} \tr \Big ( \ln (K_{aa'})  \Big ) + \frac{1}{2 N^2} J_{ab}(K_{aa'})^{-1} J_{a'b}.
\een

In other words, we could have incorporated the $J$ term into the pushforward measure $d \nu_N[t]$. This would have shifted $t_0$ obtained above by $\frac{\delta}{N^2}$ where, up to subdominant terms,

\ben
\delta = \frac{1}{2N^2} \int J_{ab} \Big (- \partial^2 + \mu + \overline{\hat{M}} \Big )^{-1} J_{ab}.  
\een

In this case, shifting in (\ref{eqnarray:asympto}) the integral over $t$ by $\delta$, we would have reobtained the expression above. \newline  

Going back to the asymptotic analysis, we have three variables, $t$, $\mu$ and $\rho$. Proceeding with the standard Laplace's method we set the variation of the exponent with respect to the first two to zero gives the following set of equations

\besn
V + A (t-t_0) & = &  0 \\
- \frac{V}{2 \lambda} + \frac{A'}{2} (t-t_0)^2 - A (t-t_0) t'_0 & = & 0.
\eesn

Of course, strictly speaking, we still need to vary with respect to $\rho$. However, this will turn out to be unnecessary for our goals as we'll see below. \newline

Therefore, obtaining $(t-t_0)$ from the first equation and plugging it into the second, we get  

\be
\label{eq:primezero}
-\frac{V}{2 \lambda} + \frac{ V^2 A'  }{2 A^2   } + t'_0  =  0 \implies - \frac{V}{2 \lambda} - V^2 (A^{-1})'  + V t'_0 = 0.
\ee

Next, let us compute 

\ben
t_0' = \frac{1}{4 \pi} \int_0^{\frac{\tilde{\Lambda}}{2}} \frac{|p|}{|p|^2 + \mu + \overline{\hat{M}} }  d|p| \sim \frac{1}{8 \pi}  \ln \bigg  ( \frac{\tilde{\Lambda}^2}{  \mu + \overline{\hat{M}}} \bigg ).
\een

Now, let us compare the three terms in (\ref{eq:primezero}). Note that the first and third terms are of the same order in $\Lambda$, while the second is greatly subdominant (the culprit being the $\frac{1}{\Lambda^4}$ in $A^{-1}$). We shall drop this term from the equation\footnote{Of course, we can't simply ignore a term in an equation. The rigorous way to proceed would be to solve the equation with this term present, and then verify that the result matches asymptotically with the one obtained by simply dropping this term.} obtaining

\ben
\frac{1}{8 \pi}  \ln \bigg  ( \frac{\tilde{\Lambda}^2}{  \mu + \overline{\hat{M}}} \bigg ) = \frac{1}{2 \lambda} \implies \mu + \overline{\hat{M}} = \tilde{\Lambda}^2  e^{- \frac{4 \pi }{ \lambda}} = \mu_0.
\een

A couple of remarks are in order. First, note that, although $\tilde{\Lambda} \to \infty$, $\mu_0$ does not diverge, as in the continuum limit $\lambda \to 0$ in just the right way for the combination above to produce a finite number. This is entirely analogous to the situation in the well-understood vector model. Second, the equation above for $\mu_0$ explains why we did not need to perform the variation with respect to $\rho$. This is because $\mu$ and $\rho$ only enter $Z[J]$ precisely through the combination $\mu_0$ above, which is already fixed by the other two variations. Note that the result is also consistent with the fact that $\mu$ was introduced as an arbitrary translation of the $\hat{M}$ variable. Any shift in the value of $\mu$ is compensated by the corresponding shift in $\hat{M}$ which in turn produces the same compensating shift in $\overline{\hat{M}}$ keeping $\mu_0$ invariant. \newline

Plugging the value of $\mu_0$ back into $Z[J]$, we see that 

\ben
Z[J] \simeq \exp \bigg [ - \frac{1}{2} \int J_{ab} \Big ( - \partial ^2 + \mu_0   \Big )^{-1}   J_{ab}  \bigg ].
\een

We have thus obtained the claimed fact that the principal chiral model in the large-$N$ limit is an infinite collection (since the $a$ and $b$ indices have an infinite range) of free theories of massive particles, all with mass $\mu_0$. \newline

\subsection*{Conclusion}

Let us summarize what has been done above. We have started with a lattice-regularized principal chiral model. After swapping the original field for a Lagrange multiplier we have looked at the large-$N$ asymptotics of the resulting expression. While the technical issues one needs to deal with can be complicated, the gist of the argument is very simple: it is well-known \cite{meckes} that Lipschitz  functions on a finite product of $O(N)$'s concentrate in the large-$N$ limit at least as well as a Gaussian. Thus assuming that the relevant function, which in our case happens to be the action, is Lipschitz, all that remains is to find the parameters of the asymptotic. This in turn is done by computing the mean and the variance since these are the parameters specifying a Gaussian.\newline

It should be emphasized that the results above relied on certain choices and assumptions. They are as follows in the same order in which they appeared: the choice of the lattice regularization, the choice of the order in which to take the limits (large-$N$ first, continuum later), the ansatz that the pushforward measure concentrates exactly as a Gaussian (this is perhaps the least justified assumption), the assumption that analyticity justifies the formal manipulations of the logarithms involving the $\hat{M}$'s and finally the assumption that the terms which were discarded by being subdominant in the continuum limit do not alter the answer. \newline

Where should one go from here? The most interesting question is whether one can obtain information on the other order of limits, which would require e.g. some uniformity in $\Lambda$ of the asymptotic results here. This would enable us to obtain a direct cross check with the extensive work which was done on the other order or limits. A related task would be to study this model numerically. In particular, it would be very interesting to see if the pushforward measure does indeed concentrate precisely as a Gaussian. Another obvious direction is to investigate to what extent the results here can be ported to more interesting models, notably Yang-Mills. It seems that one cannot do so directly in lattice Yang-Mills \cite{ym}. Nonetheless, it would be very interesting if one can transfer the methods above to loop space to study Yang-Mills in its form as a principal chiral model on the loops space \cite{gaussian}. \newline

\textbf{Acknowledgments:} The author would like to thank J. Merhej for reading a preliminary version of this paper and for the numerous comments which have greatly improved its readability.

\texttt{{\footnotesize Department of Mathematics, American University of Beirut, Beirut, Lebanon.}
}\\ \texttt{\footnotesize{Email address}} : \textbf{\footnotesize{tamer.tlas@aub.edu.lb}}

\end{document}